\documentclass{tlp}

\usepackage{times}
\usepackage[T1]{fontenc}
\usepackage{amsmath,amssymb}
\usepackage{epsfig}
\usepackage{todonotes}

\newcommand{\wrt}{{\it w.r.t. }}    
\newcommand{\eg}{\emph{e.g., }}     
\newcommand{\ie}{\emph{i.e.} }      
\newcommand{\etal}{\emph{et al.}}   

\newtheorem{definition}{Definition}

\newtheorem{theorem}{Theorem}
\newtheorem{corollary}{Corollary}
\newtheorem{proposition}{Proposition}

\begin{document}

\title[Range-based argumentation semantics as 2-valued models]
{Range-based argumentation semantics as 2-valued models}

\author[M. Osorio, J. C. Nieves]
{ MAURICIO OSORIO \\
Universidad de las Am\'ericas - Puebla\\
Depto. de Actuar\'ia, F\'isica y Matem\'aticas\\
Sta. Catarina M\'artir, Cholula, Puebla, 72820 M\'exico\\
\email{osoriomauri@gmail.com}
\and
JUAN CARLOS NIEVES\\
Department of Computing Science\\
Ume{\aa} University\\
SE-901 87, Ume{\aa}, Sweden \\
\email{jcnieves@cs.umu.se}
}

\submitted{February 28 2014}
\revised{August 28 2015 and February 3 2016}
\accepted{February 21 2016}

\maketitle

\begin{abstract}

Characterizations of semi-stable and stage extensions in terms of 2-valued logical models are presented. To this end, the so-called GL-supported and GL-stage models are defined. These two classes of logical models are logic programming counterparts of the notion of range which is an established concept in argumentation semantics.
\end{abstract}


\section{Introduction}\label{sec:intro}

Argumentation has been regarded as a non-monotonic reasoning approach since it was suggested as an inference reasoning approach \cite{PV02}. Dung showed that argumentation inference can be regarded as a logic programming inference with \emph{negation as failure }\cite{Dung95}. In his seminal paper \cite{Dung95}, Dung introduced four argumentation semantics: \emph{grounded, stable, preferred} and \emph{complete} semantics.
Currently, it is known that these four argumentation semantics introduced by Dung can be regarded as logic programming inferences by using different mappings, from argumentation frameworks (AFs) into logic programs, and different logic programming semantics (see Section \ref{sec:relatedwork}).

Following Dung's argumentation style, several new argumentation semantics have been proposed. Among them, ideal, semi-stable, stage and CF2 have been deeply explored \cite{BaroniCG11}. Semi-stable and stage semantics were introduced from different points of view; however,
they have been defined in terms of the so-called \emph{ranges} of complete extensions and conflict-free sets, respectively.  It seems that by using the concept of range, one can define different classes of argumentation semantics as is the case with the semi-stable and stage semantics.
Given that the concept of range seems a fundamental component of definitions of argumentation semantics such as semi-stable and stage semantics, the following question arises:

\begin{quote}[Q1]
\emph{How the concept of range can be captured from the logic programming point of view?}
\end{quote}

This question takes relevance in the understanding of argumentation as logic programming.

In this paper, we argue that for capturing the idea of range from the logic programming point view, logic programming reductions which have been used for defining logic programming semantics such as stable model \cite{GelLif88} and p-stable \cite{OsoNavArrBor06} semantics are important.
To show this, we introduce a general schema $SC_1$ which takes as input a logic program $P$ and a set of atoms $M$, then considering a function $R$ which maps $P$ into another logic program, $SC_1$ returns a subset of atoms of the signature of $P$\footnote{The formal definition of $SC_1$ is presented in Section \ref{sec:idealmodels}}.
In order to infer ranges from the argumentation point of view using $SC_1$, the logic program $P$ has to capture an argumentation framework. Let us observe that there are different mappings from AFs into logic programs which have been used for characterizing Dung's argumentation semantics as logic programming inferences \cite{CarNivOso09,Dung95,CamSaAlc13,Strass13}. In this sense, the following question arises:

\begin{quote}[Q2]
\emph{Can the mappings used for characterizing Dung's argumentation semantics characterize range-based argumentation semantics using $SC_1$? }
\end{quote}

In order to give an answer to Q2, we consider the mappings $\Pi_{AF}$  and $\Pi_{AF}^-$ which have been used for characterizing Dung's argumentation semantics in terms of logic programming semantics. $\Pi_{AF}$ has been shown to be a flexible mapping to characterize the grounded, stable, preferred, complete and ideal semantics by using logic programming semantics such as, the well-founded, stable, p-stable, Clark's completion and well-founded$^+$ semantics, respectively \cite{CarNivOso09,NivOsoCor08,NivOso14}. $\Pi_{AF}^-$  has been used to characterize the grounded, stable, preferred, complete, semi-stable and CF2 \cite{Dung95,NivOsoZep11,Strass13}.

Considering $\Pi_{AF}$  and $\Pi_{AF}^-$ for defining two different instantiations of $SC_1$, we will define the so-called \emph{GL-supported} and \emph{GL-stage} models. We will show that \emph{GL-supported} and \emph{GL-stage} models characterize the semi-stable and stage extensions, respectively. In these instantiations of $SC_1$, we will instantiate the function $R$ with the well-known Gelfond-Lifschitz reduction which is the core of the construction of the \emph{stable model semantics} \cite{Gel08}. Moreover, we will point out that $R$ can be instantiated with the $RED$ reduction, which is the core of the \emph{p-stable semantics} \cite{OsoNavArrBor06}, getting the same effect in the constructions of the \emph{GL-supported} and \emph{GL-stage} models.

To the best of our knowledge, $SC_1$ is the first schema designed to capture the range concept from a logic programming point of view. It is worth mentioning that a range-based semantics as semi-stable semantics has been already characterized as logic programming inference \cite{CamSaAlc13,Strass13}; however, these characterizations do not offer a schema for capturing the concept of range from a logic programming point of view in order to characterize (or construct) other range-based argumentation semantics such as stage semantics.

The rest of the paper is structured as follows: In Section \ref{sec:back}, a basic background about logic programming and argumentation is introduced. In Section \ref{sec:Semi-Stable-Stage-Semantics}, by considering a couple of instantiations of $SC_1$, we introduce the so-called \emph{GL-supported} and \emph{GL-stage} models; moreover, we show how these models characterize both semi-stable and stage extensions. In Section \ref{sec:relatedwork}, a discussion of related work is presented. In the last section, our conclusions are presented.

\section{Background}\label{sec:back}
In this section, we introduce the syntax of normal logic programs and the p-stable and stable model semantics. After this, some basic concepts of argumentation theory are presented. At the end of the section, the mappings $\Pi_{AF}$ and $\Pi_{AF}^-$ are introduced.

\subsection{Logic Programs: Syntax}

A signature ${\cal L}$ is a finite set of elements that we call atoms. A literal is an atom $a$ (called \emph{a positive literal}), or the negation of an atom $not~
a$ (called \emph{a negative literal}). Given a set of atoms $\{ a_1, \dots, a_n \}$, we write $not ~ \{
a_1, \dots, a_n \}$ to denote the set of literals $\{not~ a_1,
\dots, not ~ a_n\}$. A \emph{normal clause} $C$ is written as:
$$a_0 \leftarrow a_1, \dots, a_{j}, not~ a_{j+1}, \dots, not~a_n$$

\noindent where $a_i$ is an atom, $0 \leq i \leq n$. When $n = 0$ the normal
clause is called a fact and is an abbreviation of $a_0 \leftarrow \top$, where $\top$
is the ever true atom. A \emph{normal logic program} is a finite set of normal clauses.
Sometimes, we denote a clause {\it C} by $ a \leftarrow {\cal B}^+ ,
\;not\; {\cal B}^-$, where ${\cal B}^+$ contains all the positive body
literals and ${\cal B}^-$ contains all the negative body literals.
When ${\cal B}^- = \emptyset$, the clause {\it C} is called a \emph{definite
clause}. A \emph{definite program} is a finite set of
definite clauses. ${\cal L}_P$ denotes the set of atoms that occurs in P. Given a signature ${\cal L}$, $Prog_{{\cal L}}$ denotes the set of all the programs
defined over ${\cal L}$. Given a normal logic program $P$, $Facts(P) = \{ a | a\leftarrow \top \in P \}$.

In some cases we treat a logic program as a logical theory. In these cases, each negative literal $not \; a$ is replaced by $\neg a$ where $\neg$ is regarded as the classical negation in classical logic. Logical consequence in classical logic is denoted by $\vdash$. Given a
set of proposition symbols $S$ and a logical theory (a set of well-formed formulae) $\Gamma$, $\Gamma \vdash S$ if $\forall s \in S$ $\Gamma \vdash s$.

Given a normal logic program P, a set of atoms is a classical model of $P$ if the induced interpretation evaluates $P$ to true. If $M \subseteq {\cal L}_P$, we write $P \Vdash M$ when: $P \vdash M$ and $M$ is a classical 2-valued model of the logical theory obtained from $P$ (\ie atoms in $M$ are set to true, and atoms not in $M$ to false). We say that a model $M$ of a program $P$ is minimal  if a model $M'$ of $P$ different from $M$ such that $M' \subset M$ does not exist.

\subsection{Stable model and p-stable semantics}\label{subsec:Pstable}
\emph{Stable model semantics} is one of the most influential logic programming semantics in the non-monotonic reasoning community \cite{Bar03} and is defined as follows:

\begin{definition}\cite{GelLif88}\label{def:stable}
Let \emph{P} be a normal logic program. For any set $S \subseteq {\cal L}_P$, let $P^S$ be the definite logic program obtained from \emph{P} by
deleting

\begin{description}
\item[(i)] each clause that has a formula $not \; l$ in its body with
$l \in S$, and then

\item[(ii)] all formul\ae\ of the form $not \; l$ in the bodies of
the remaining rules.
\end{description}

\noindent Then, \emph{S} is a stable model of \emph{P} if \emph{S} is a minimal model of $P^S$.
$Stable(P)$ denotes the set of stable models of $P$
\end{definition}

From hereon, whenever we say \emph{Gelfond-Lifschitz (GL) reduction}, we mean the reduction $P^{S}$. As we can observe GL reduction is the core of the stable model semantics.

There is an extension of the stable model semantics which is called \emph{p-stable semantics} \cite{OsoNavArrBor06}. P-stable semantics was formulated in terms of \emph{Paraconsistent logics}.
Like stable model semantics, p-stable semantics is defined in terms of a single reduction, \emph{RED}, which is defined as fo\-llows\-:

\begin{definition}\cite{OsoNavArrBor06}\label{def:RED-normalProgram}
Let $P$ be a normal program and M be a set of atoms. We define
$RED(P,M) := \{ l \leftarrow {\cal B}^+ , \ not \; ({\cal B}^- \cap M  )| l \leftarrow {\cal B}^+ , \ not \; {\cal B}^- \in P \}$.
\end{definition}

As we can see, GL reduction and $RED$ reduction have different behaviors. On the one hand, the output of GL reduction always is a definite program; on the other hand, the output of RED reduction can contain normal clauses.

By considering $RED$ reduction, the \emph{p-stable} semantics for normal logic programs is defined as follows:

\begin{definition}\cite{OsoNavArrBor06}\label{def:pstable}
Let $P$ be a normal program and $M$ be a set of atoms. We say that $M$ is a p-stable model of $P$ if $RED(P,M ) \Vdash M$. $P\text{-}stable(P)$ denotes the set of p-stable models of $P$.
\end{definition}

The stable model and p-stable semantics are two particular 2-valued semantics for normal program.
In general terms, a logic programming semantics $SEM$ is a function from the class of all programs into the powerset of the set of (2-valued) models.

Before moving on, let us introduce the following notation. Let $P$ be a logic program, $2SEM(P)$ denotes the 2-valued models of $P$.
Given two logic programming semantics $SEM_1$ and $SEM_2$, $SEM_1$ is \emph{stronger than} $SEM_2$ if for every logic program $P$,  $SEM_1(P)  \subseteq SEM_2(P)$. Let us observe that the relation \emph{stronger than} between logic programming semantics is basically defining an order between logic programming semantics.

\subsection{Argumentation theory}\label{sec:back-argumentation}

In this section, we introduce the definition of some argumentation semantics. To this end, we start by defining the basic structure of an argumentation framework (AF).

\begin{definition}\cite{Dung95}
An argumentation framework is a pair $AF := \langle AR, attacks
\rangle$, where \emph{AR} is a finite set of arguments, and
\emph{attacks} is a binary relation on \emph{AR}, \ie \emph{attacks}
$\subseteq AR \times AR$.
\end{definition}

We say that \emph{a attacks b} (or $b$ is attacked by $a$) if $(a,b) \in attacks$ holds. Similarly, we say that a set $S$ of arguments attacks $b$ (or $b$ is attacked by $S$) if $b$ is attacked by an argument in $S$. We say that $c$ \emph{defends} $a$ if $(b,a)$ and $(c,b)$ belongs to $attacks$.

Let us observe that an AF is a simple structure which captures the conflicts of a given set of arguments. In order to select \emph{coherent points of view} from a set of conflicts of arguments, Dung introduced the so-called \emph{argumentation semantics}. These argumentation semantics are based on the concept of an \emph{admissible set}.

\begin{definition}\cite{Dung95}
\begin{itemize}
\item A set {\it S} of arguments is said to be conflict-free if there are
no arguments {\it a, b} in {\it S} such that {\it a attacks b}.
\item  An argument $a \in AR$ is acceptable with respect to a set $S$ of
arguments if for each argument $b \in AR$: If $b$ attacks $a$ then $b$ is attacked by $S$.
\item  A conflict-free set of arguments $S$ is \emph{admissible} if each argument in $S$ is acceptable \wrt $S$.
\end{itemize}
\end{definition}

Let us introduce some notation. Let $AF := \langle AR, attacks \rangle$ be an AF and $S \subseteq AR$. $S^+ = \{b  | a \in S \text{ and } (a,b) \in attacks\}$.

\begin{definition}\cite{Cam06,Dung95}\label{def:semArg}
Let $AF := \langle AR, attacks \rangle$ be an argumentation framework. An admissible set of arguments $S \subseteq AR$ is:
\begin{itemize}
  \item \emph{stable} if $S$ attacks each argument which does not belong to $S$.
  \item \emph{preferred} if $S$ is a maximal (\wrt set inclusion) admissible set of $AF$ .
  \item \emph{complete} if each argument, which is acceptable with respect to $S$, belongs to $S$.
  \item \emph{semi-stable} if $S$ is a complete extension such that $S \cup S^+$ is maximal \wrt set inclusion.
\end{itemize}
\end{definition}

In addition to argumentation semantics based on admissible sets, there are other approaches for defining argumentation semantics \cite{BaroniCG11}. One of these approaches is the approach based on \emph{conflict-free sets} \cite{Verheij96}. Considering conflict-free sets, Verheij introduced the so-called \emph{stage semantics}:

\begin{definition}\label{def:StageSem}
Let $AF :=\langle AR, attacks \rangle $ be an argumentation framework. $E$ is a stage extension if $E$ is a conflict-free set and $E \cup E^+$ is maximal \wrt set inclusion.
\end{definition}

Let us observe that both semi-stable and stage semantics are based on the so-called \emph{range} which is defined as follows: If $E$ is a set of arguments, then $E \cup E^+$ is called its range. According to the literature, the notion of range was first introduced by Verheij \cite{Verheij96}.

\subsection{Mappings from argumentation frameworks to normal programs}\label{sec:mapping}

In this section, two mappings from an AF into a logic program will be presented.
These mappings are based on the ideas of \emph{conflictfreeness} and \emph{reinstatement} which are the basic concepts behind the definitions of conflict-free sets and admissible sets. In these mappings, the predicate $def(x)$ is used, with the intended meaning of $def(x)$ being ``$x$ is a defeated argument''.

A pair of mapping functions \wrt an argument is defined as follows.

\begin{definition}\label{mapping}
Let  $AF := \langle AR, attacks \rangle $ be an argumentation framework and $a \in AR$. We define a pair of mappings functions:

$$ \Pi^{-}(a)  =  \bigcup_{b :(b,a) \in attacks } \{ def(a) \leftarrow  \; not \; def(b)\}$$

$$\Pi^{+}(a)  = \bigcup_{b :(b,a) \in attacks} \{ def(a) \leftarrow \bigwedge_{c:(c,b) \in attacks} def(c) \}$$

\end{definition}

Let us observe that $ \Pi^{-}(a)$ suggests that an argument $a$ is defeated when anyone of the arguments which attack $a$ is not defeated. In other words, an argument that has an attacker that is not defeated has to be defeated; hence, $\Pi^-(a)$ stands for conflictfreeness. $\Pi^{+}(a)$ suggests that an argument $a$ is defeated when all the arguments that defends $a$ are defeated. In other word, any argument that is not defeated has to be defended; therefore $\Pi^{+}(a)$ stands for admissibility.

One can see that if a given argument $a$ has no attacks, then $\Pi^{-}(a) = \{\}$ and $\Pi^{+}(a) = \{\}$.
This situation happens because an argument that has no attacks is an acceptable argument which means that it belongs to all extensions sets of an $AF$.

By considering $\Pi^{-}(a)$ and $\Pi^{+}(a)$, two mappings from an AF into a logic program are introduced.

\begin{definition}\label{def:mappings}
Let $AF :=\langle AR, attacks \rangle $ be an argumentation
framework. We define their associated normal programs as follows:

$$\Pi_{AF}^- := \bigcup_{a \in AR} \{ \Pi^-(a) \}$$
$$\Pi_{AF} := \Pi_{AF}^- \cup \bigcup_{a \in AR} \{ \Pi^+(a) \}  $$
\end{definition}

Observing Definition \ref{def:mappings}, it is obvious that $\Pi_{AF}^-$ is a subset of  $\Pi_{AF}$. However, each mapping is capturing different concepts: $\Pi_{AF}^-$ is a declarative specification of the idea of conflictfreeness and $\Pi_{AF}$ is a declarative specification of both ideas: conflictfreeness and reinstatement. Indeed, one can see that the 2-valued logical models of $\Pi_{AF}^-$ characterize the conflict-free sets of  an $AF$ and the 2-valued logical models of $\Pi_{AF}$ characterize the admissible sets of an $AF$.

\section{Semi-stable and Stage extensions as 2-valued models}\label{sec:Semi-Stable-Stage-Semantics}

This section introduces the main results of this paper. In particular, we will show that the following schema $SC_1$ suggests an interpretation of range from the logic programming point of view:

\[SC_1(P,M) = Facts(R(P,M)) \cup \{{\cal L}_{P} \setminus M\}\]

\noindent in which $P$ is a logic program, $R$ is a function which maps a logic program into another logic program considering a set of atoms $M \subseteq {\cal L}_{P}$.

In order to show our results, we will introduce two instantiations of the schema $SC_1$. These instantiations will lead to the so-called \emph{GL-supported models} and  \emph{GL-stage models}. We will show that the GL-supported models of $\Pi_{AF}$ characterize the semi-stable extensions of a given AF (Theorem \ref{theorem:semi-models-semi-models}); moreover, the GL-stage models of $\Pi_{AF}^-$ characterize the stage extensions of a given AF (Theorem \ref{theorem:stage-models}).

\subsection{Semi-Stable Semantics}\label{sec:idealmodels}

We start presenting our results \wrt semi-stable semantics. To this end, let us start defining the concept of a \emph{supported model}.

\begin{definition}[Supported model]\label{def:supported-models}
Let $P$ be a logic program and $M$ be a 2-valued model of $P$. $M$ is \emph{a supported model} of $P$ if for each $a \in M$, there is $ a_0 \leftarrow {\cal B}^+, \;not\; {\cal B}^- \in P$ such that $a = a_0$,  ${\cal B}^+ \subseteq M$ and  ${\cal B}^- \cap M = \emptyset$.
\end{definition}

As we saw in Definition \ref{def:semArg}, semi-stable extensions are defined in terms of complete extensions. It has been shown that the supported models of $\Pi_{AF}$ characterize the complete extensions of a given AF \cite{OsoNivSan13ENC}. By having in mind this result, we introduce an instantiation of the schema $SC_1$ in order to define the concept of \emph{GL-supported-model}.

\begin{definition}[GL-supported-model]\label{def:range-supported-models}
Let $AF = \langle AR, Attacks \rangle$ be an argumentation framework and $M$ be a supported model of $\Pi_{AF}$.
$M$ is a GL-supported-model of $\Pi_{AF}$ if $Facts((\Pi_{AF})^M) \cup \{ {\cal L}_{\Pi_{AF}} \setminus M \}$ is maximal \wrt set inclusion. $GLModels(\Pi_{AF})$ denotes the GL-supported models of $\Pi_{AF}$.
\end{definition}

In other words, a supported model $M$  of $\Pi_{AF}$ is a GL-supported-model if for every supported model $N$ of $\Pi_{AF}$ such that $N$ is different of $M$, $SC(\Pi_{AF}, M) \not\subset SC(\Pi_{AF}, N)$ where $SC(\Pi_{AF}, X) = Facts((\Pi_{AF})^X) \cup \{ {\cal L}_{\Pi_{AF}} \setminus X \}$.

Let us observe that the function $R$ of the schema $SC_1$ was replaced by the GL-reduction in the construction of a GL-supported model. One of the main constructions of the definition of a GL-supported model is $Facts((\Pi_{AF})^M)$. This part of the construction of a GL-supported model is basically  characterizing the set $E^+$ where $E$ is a complete extension. We can see that the GL reduction is quite important for this construction. As we saw in Definition \ref{def:stable}, GL reduction is the core of the definition of stable models.

We want to point out that the definition of GL-supported models can also be based on the RED reduction which is the reduction used for defining p-stable models (see Definition \ref{def:pstable}).
This similarity between RED and GL reductions argues that both RED and GL reductions can play an important role for capturing the idea of range of an argumentation framework from a logic programming point of view.  As we will see in the following theorem, GL-supported models characterize semi-stable extensions; hence, both RED and GL reductions play an important role for capturing semi-stable extension as 2-valued logical models.

In order to simplify the presentation of some results, let us introduce the following notation. Let $E_M =  \{x | def(x) \in {\cal L}_{\Pi_{AF}} \setminus M\}$ and $E_M^+ = \{ x | def(x) \in Facts((\Pi_{AF})^M) \}$ where $M \subseteq {\cal L}_{\Pi_{AF}}$. As we can see, $E_M$ and $E_M^+$ are basically sets of arguments which are induced by a set of atoms $M$.

\begin{theorem}\label{theorem:semi-models-semi-models}
Let $AF = \langle AR, attacks \rangle$ be an argumentation framework and $M \subseteq {\cal L}_{\Pi_{AF}}$. $M$ is a GL-supported model of $\Pi_{AF}$ iff $E_M$ is a semi-stable extension of $AF$.
\end{theorem}
\begin{proof}
Let us start introducing the following result from \cite{OsoNivSan13ENC}:
\begin{description}
  \item [$\mathbf{R1}$:] Let $AF = \langle AR, attacks \rangle$ be an argumentation framework. $M$ is a supported model of $\Pi_{AF}$ iff $E_M$ is a complete extension of $\Pi_{AF}$.
\end{description}

The proof goes as follows:
\begin{description}
  \item[=>] Let $M$ be a GL-supported model of $\Pi_{AF}$ and $M^* = Facts((\Pi_{AF})^M)$. Then by definition of a GL-supported model, $M^* \cup {\cal L}_{\Pi_{AF}} \setminus M$ is maximal \wrt set inclusion.  Moreover, $M$ is a supported model. Therefore, by $\mathbf{R1}$, $E_M$ is a complete extension. Hence, it is not hard to see that  $E_M \cup E_M^+$ is a range with respect to the complete extension $E_M$. Since $M^* \cup {\cal L}_{\Pi_{AF}} \setminus M$ is maximal \wrt set inclusion, $E_M \cup E_M^+$ is also maximal \wrt set inclusion. Hence, $E_M$ is a semi-stable extension.
  \item[<=] Let us suppose that $E$ is a semi-stable extension of $AF$. By definition $E \cup E^+$ is maximal \wrt set inclusion and $E$ is a complete extension. By $\mathbf{R1}$,  there exists a supported model $M$ of $\Pi_{AF}$ such that $E_M = E$; moreover, any supported model $N$ of $\Pi_{AF}$ has the property that $E_N \cup E_N^+$ is maximal \wrt set inclusion. Then ${\cal L}_{\Pi_{AF}} \setminus N \cup Facts((\Pi_{AF})^N)$ is maximal \wrt set inclusion. Then $N$ is a GL-supported model of $\Pi_{AF}$.
\end{description}
\end{proof}

An interesting property of GL-supported models is that they can be characterized by both the set of p-stable models of $\Pi_{AF}$ and the set of 2-valued models of $\Pi_{AF}$.

\begin{proposition}\label{prop:CharacterizationSemi-model}
Let $AF = \langle AR, attacks \rangle$ be an argumentation framework.
\begin{enumerate}
\item $M$ is a GL-supported model of $\Pi_{AF}$ iff $E_M \cup E_M^+$ is maximal \wrt set inclusion where $M$ is a 2-valued model of $\Pi_{AF}$.
  \item $M$ is a GL-supported model of $\Pi_{AF}$ iff $E_M \cup E_M^+$ is maximal \wrt set inclusion where $M$ is a p-stable model $\Pi_{AF}$.
\end{enumerate}
\end{proposition}
\begin{proof}
We start introducing the following observations from the state of the art:

\begin{enumerate}
  \item Let $M \subseteq {\cal L}_{\Pi_{AF}}$. $M$ is a 2-valued model of $\Pi_{AF}$ iff  $E_M$ is an admissible extension of $AF$ \cite{NivOso14}.
  \item According to Proposition 4 by \cite{CaminadaCD12} the following statements are equivalent:
      \begin{enumerate}
        \item $E$ is a complete extension such that $E \cup E^+$ is maximal (\wrt set inclusion).
        \item $E$ is an admissible set such that $E \cup E^+$ is maximal (\wrt set inclusion).
      \end{enumerate}
  \item Let $M \subseteq {\cal L}_{\Pi_{AF}}$. $M$ is a p-stable model of $\Pi_{AF}$ iff $E_M$ is a preferred  extension of $AF$ \cite{CarNivOso09}.
\end{enumerate}

Now let us prove each of the points of the proposition:

\begin{enumerate}
  \item  $M$ is a GL-supported model of $\Pi_{AF}$ iff $Facts((\Pi_{AF})^M) \cup \{ {\cal L}_{\Pi_{AF}} \setminus M \}$ is maximal \wrt set inclusion and $M$ is a supported model. By Theorem \ref{theorem:semi-models-semi-models}, $E_M \cup E_M^+$ is maximal \wrt set inclusion and $M$ is a supported model iff $E_M \cup E_M^+$ is maximal and $E_M$ is a complete extension of $AF$. By Observation 2, $E_M \cup E_M^+$ is maximal and $E_M$ is a complete extension of $AF$ iff $E_M \cup E_M^+$ is maximal and $E_M$ is an admissible extension of $AF$. Hence, the result follows by Observation 1 which argues that any 2-valued model of $\Pi_{AF}$ characterizes an admissible set of $AF$.
  \item Let us start by observing that semi-stable extensions can be characterized by preferred extensions with maximal range which means: $E$ is a semi-stable extension iff $E \cup E^+$ is maximal (\wrt set inclusion) and $E$ is a preferred extension (see Proposition 13 from \cite{BaroniCG11}). Hence, the result follows by Observation 3 and Theorem \ref{theorem:semi-models-semi-models}.
\end{enumerate}
\end{proof}

A direct consequence of Proposition \ref{prop:CharacterizationSemi-model} and Theorem \ref{theorem:semi-models-semi-models} is the following corollary which introduces a pair of characterizations of semi-stable extensions as 2-valued models and p-stable models of $\Pi_{AF}$.

\begin{corollary}\label{corol:bound}
Let $AF = \langle AR, attacks \rangle$ be an argumentation framework. 
\begin{enumerate}
  \item Let $M$ be a p-stable model of $\Pi_{AF}$. $E_M$ is a semi-stable extension of $AF$ iff $E_M \cup E_M^+$ is maximal \wrt set inclusion.
\item Let $M$ be a 2-valued model of $\Pi_{AF}$. $E_M$ is a semi-stable extension of $AF$ iff $E_M \cup E_M^+$ is maximal \wrt set inclusion.
\end{enumerate}
\end{corollary}

Observing Corollary \ref{corol:bound}, we can see that there is an interval of logic programming semantics which characterizes semi-stable extensions. This interval of logic programming semantics is defined by the order-relation between logic programming semantics: \emph{stronger than}.
This result is formalized by the following corollary.

\begin{corollary}\label{corol:Interval-Semi-stable}
Let $AF = \langle AR, attacks \rangle$ be an argumentation framework and $SEM$ be a logic programming semantics such that $SEM$ is stronger than $2SEM$ and $P\text{-}stable$ is stronger than $SEM$. If $M \in SEM(\Pi_{AF})$, then $E_M$ is a semi-stable extension of $AF$ iff $E_M \cup E_M^+$ is maximal \wrt set inclusion.
\end{corollary}

Given the relation of semi-stable extensions with the stable and preferred extensions, we can observe some relations between GL-supported models \wrt the stable model semantics \cite{GelLif88} and p-stable semantics.

\begin{proposition}
Let $AF = \langle AR, attacks \rangle$ be an argumentation framework.

\begin{enumerate}
  \item If $M$ is a stable model of $\Pi_{AF}$ then  $M$ is a GL-supported model of $\Pi_{AF}$.
  \item If $M$ is a GL-supported model of $\Pi_{AF}$ then  $M$ is a p-stable model of $\Pi_{AF}$.
\end{enumerate}
\end{proposition}
\begin{proof}
\begin{enumerate}
  \item It follows from Theorem \ref{theorem:semi-models-semi-models} and Theorem 2 by \cite{CaminadaCD12}.
  \item It follows from Theorem \ref{theorem:semi-models-semi-models} and Theorem 3 by \cite{CaminadaCD12}.
\end{enumerate}
\end{proof}

\begin{proposition}
Let $AF = \langle AR, attacks \rangle$ be an argumentation framework such that $Stable(\Pi_{AF}) \neq \emptyset$. Then, $Stable(\Pi_{AF}) = GLModels(\Pi_{AF})$.
\end{proposition}
\begin{proof}
We know that $E$ is a stable extension of $AF$  iff $E = E_M$  where  $M$  is a stable model of $\Pi_{AF}$ (Theorem 5 by \cite{CarNivOso09}). Hence, the result follows from Theorem \ref{theorem:semi-models-semi-models} and Theorem 5 by  \cite{CaminadaCD12}.
\end{proof}

\subsection{Stage Semantics}
We have seen that the idea of range \wrt complete extensions can be captured by instantiating the schema $SC_1$ considering supported-models, the GL-reduction and $\Pi_{AF}$.

In Section \ref{sec:mapping}, the mappings $\Pi_{AF}^-$ and  $\Pi_{AF}$ were introduced. We have observed that $\Pi_{AF}^-$ is basically a declarative specification of conflict-free sets.
Given that stage semantics is based on conflict-free sets, we will consider  $\Pi_{AF}^-$ for instantiating $SC_1$ and defining the so-called \emph{GL-stage models}:

\begin{definition}
Let $AF = \langle AR, attacks \rangle$ be an argumentation framework and $M$ be a 2-valued model of $\Pi_{AF}^-$. $M$ is a GL-stage model of $\Pi_{AF}^-$ if
$Facts((\Pi_{AF}^-)^M) \cup \{ {\cal L}_{\Pi_{AF}^-} \setminus M \}$
is maximal \wrt set inclusion.
\end{definition}

In other words, a 2-valued model $M$  of $\Pi_{AF}^-$ is a GL-stage-model if for every 2-valued model $N$ of $\Pi_{AF}^-$ such that $N$ is different of $M$, $SC'(\Pi_{AF}^-, M) \not\subset SC'(\Pi_{AF}^-, N)$ where $SC'(\Pi_{AF}^-, X) = Facts((\Pi_{AF}^-)^X) \cup \{ {\cal L}_{\Pi_{AF}^-} \setminus X \}$.

In this characterization of $SC_1$, once again we are replacing the function $R$ of  $SC_1$ by the GL-reduction; however, one can use RED reduction for defining GL-stage models.

One can observe that GL-stage models characterize stage extensions. In order to formalize this result, the following notation is introduced: Let $E'_M =  \{x | def(x) \in {\cal L}_{\Pi_{AF}^-} \setminus M\}$ and $E_M^{'+} = \{ x | def(x) \in Facts((\Pi_{AF}^-)^M) \}$ where $M \subseteq {\cal L}_{\Pi_{AF}^-}$. Like $E_M$ and $E_M^{+}$, $E'_M$ and $E_M^{'+}$ return sets of arguments given a set of atoms $M$  from $\Pi_{AF}^-$.

\begin{theorem}\label{theorem:stage-models}
Let $AF :=\langle AR, attacks \rangle $ be an argumentation framework. $M$ is a GL-stage model of $\Pi_{AF}^-$ iff $E'_M $ is a stage extension of $AF$.
\end{theorem}
\begin{proof}
Let us start with one observation:
\begin{description}
  \item[$\textbf{O1}$:] Let $AF = \langle AR, attacks \rangle$ be an argumentation framework.  $E'_M$ is a conflict-free set of $AF$ iff $M$ is a 2-valued model of  $\Pi_{AF}^-$.
\end{description}

\begin{description}
  \item[=>] Let $M$ be a GL-stage model of $\Pi_{AF}^-$ and $M^* = Facts((\Pi_{AF}^-)^M)$. Then by definition of a GL-stage model, $M^* \cup {\cal L}_{\Pi_{AF}^-} \setminus M$ is maximal \wrt set inclusion and $M$ is a 2-valued model. Hence, by $\mathbf{O1}$, $E'_M$ is a conflict-free set. One can see that  $E'_M \cup E_M^{'+}$ is a range with respect to the conflict-free set $E'_M$. Since $M^* \cup {\cal L}_{\Pi_{AF}^-} \setminus M$ is maximal \wrt set inclusion, $E'_M \cup E_M^{'+}$ is also maximal \wrt set inclusion. Hence, $E_M$ is a stage extension.
  \item[<=] Let us suppose that $E$ is a stage extension of $AF$. By definition $E \cup E^+$ is maximal \wrt set inclusion and $E$ is a conflict-free set. By $\mathbf{O1}$,  there exists a 2-valued model $M$ of $\Pi_{AF}^-$ such that $E'_M = E$; moreover, any 2-valued model $N$ of $\Pi_{AF}^-$ has the property that $E'_N \cup E_N^{'+}$ is maximal \wrt set inclusion. Then ${\cal L}_{\Pi_{AF}^-} \setminus N \cup Facts((\Pi_{AF}^-)^N)$ is maximal \wrt set inclusion. Then $N$ is a GL-stage model of $\Pi_{AF}^-$.
\end{description}

\end{proof}

Let us observe that $Facts((\Pi_{AF}^-)^M) \cup \{ {\cal L}_{\Pi_{AF}^-} \setminus M \}$, which is the key construction of GL-stable models, is basically characterizing ranges \wrt conflict-free sets.

Dvor{\'{a}}k and Woltran have shown that the decision problems of the credulous and sceptical inferences are of complexity $\Sigma_2^P$-hard  and  $\Pi_2^p$-hard, respectively, for both semi-stable and stage semantics \cite{DvorakW10}. Hence it is straightforward to observe that the decision problems of the credulous and sceptical inferences are of complexity $\Sigma_2^P$-hard  and  $\Pi_2^p$-hard, respectively, for both GL-supported models and GL-stage models. Let us remember that GL-supported models and GL-stage models are defined under the resulting class of programs of the mappings $\Pi_{AF}$ and $\Pi_{AF}^-$, respectively.

\section{Related work}\label{sec:relatedwork}

Dung showed that argumentation can be viewed as logic programming with \emph{negation as failure} and vice versa. This strong relationship between argumentation and logic programming has given way to intensive research in order to explore the relationship between argumentation and logic programming \cite{CamSaAlc13,CarNivOso09,Dung95,NivOsoZepCor05,NivOsoCor08,OsoNivSan13ENC,NivOsoZep11,Strass13,WuCG09}.
A basic requirement for exploring the relationship between argumentation and logic programming is to identify proper mappings which allow us to transform an argumentation framework into a logic program and vice versa. The flexibility of these mappings will frame the understanding of argumentation as logic programming (and vice versa). Therefore, defining simple and flexible mappings which regard argumentation as logic programming (and vice versa) will impact the use of logic programming in argumentation (and vice versa). Currently, we can find different mappings for regarding argumentation as logic programming (and vice versa) \cite{CamSaAlc13,CarNivOso09,Dung95,GabbayG09}. All of them offer different interpretations of argumentation as logic programming (and vice versa). Depending on these interpretations, one can identify direct relationships between argumentation inferences and logic programming inferences.

In this paper, we have limited our attention to the interpretation of argumentation as logic programming. In this sense, there are some characterizations of semi-stable inference as logic programming inference \cite{CamSaAlc13,Strass13}. Caminada \etal, \cite{CamSaAlc13}, showed that the semi-stable semantics can be characterized by the \emph{L-stable semantics} and the mapping $P_{AF}$ which is defined as follows: Given an argumentation framework $AF := \langle AR, attacks
\rangle$:

$$P_{AF} = \bigcup_{x\in AR} \{ x \leftarrow \bigwedge_{(y,x) \in attacks} not\; y \}$$

Unlike GL-supported models which are 2-valued models, the models of the L-stable semantics are 3-valued. Moreover, unlike $\Pi_{AF}$ which is a declarative specification of admissible sets, $P_{AF}$ is a declarative specification of conflict-free sets.

Strass \cite{Strass13} has also showed that the semi-stable semantics can by characterized by both the so-called \emph{L-supported models} and \emph{L-Stable models}. Unlike Caminada's characterization and our characterizations, Strass considered the mapping $\Pi_{AF}^-$. As we have observed in Section \ref{sec:mapping}, the clauses of $\Pi_{AF}^-$ are a subset of $\Pi_{AF}$ which is the mapping that we considered in both Theorem \ref{theorem:semi-models-semi-models} and Corollary \ref{corol:Interval-Semi-stable}.  It is worth mentioning that the mapping introduced by Dung \cite{Dung95} can be transformed into $\Pi_{AF}^-$.

We cannot argue that one characterization is better than the other; however, we can observe that all these characterizations, including the ones introduced in this paper, offer different interpretations of semi-stable inference. Moreover, given that semi-stable inference has been characterized in terms of  both L-stable semantics and L-supported modes, it seems that these logic programming semantics are related to GL-supported semantics.

In the literature, there are different characterizations of argumentation semantics in terms of logic programming semantics. A summary of these characterization is presented in Table \ref{table:summary}.

\begin{table}\caption{Characterization of argumentation semantics as logic programming inferences.}
\begin{tabular}{|p{2.5cm}|p{2.5cm}|p{2.5cm}|p{2.5cm}|}\hline
\textbf{Argumentation semantics} & \textbf{Logic programming semantics using $P_{AF}$} & \textbf{Logic programming semantics using $\Pi_{AF}$}  & \textbf{Logic programming semantics using $\Pi_{AF}^-$}\\ \hline
Grounded Semantics   &  Well-founded semantics \cite{CamSaAlc13}, the Kripke-Kleene model \cite{Strass13}      &  Well-founded semantics \cite{CarNivOso09}   & Well-founded semantics \cite{Dung95}, the Kripke-Kleene model \cite{Strass13}  \\ \hline
Stable Semantics   &  Stable model semantics \cite{CamSaAlc13,NivOsoZepCor05}, Supported models \cite{Strass13}  & Stable model semantics \cite{CarNivOso09}  & Stable models semantics \cite{Dung95}, Supported models \cite{Strass13} \\ \hline
Preferred Semantics   &  Regular semantics \cite{CamSaAlc13}, M-supported models, M-stable models \cite{Strass13}      & P-stable Semantics  \cite{CarNivOso09}&  M-supported models, M-stable models \cite{Strass13}\\ \hline
Complete Semantics   &  3-valued stable semantics  \cite{WuCG09,Strass13}, 3-valued supported models  \cite{Strass13}  & Supported Models \cite{OsoNivSan13ENC} & 3-valued stable semantics, 3-valued supported models  \cite{Strass13} \\ \hline
Semi-stable Semantics   &  L-Stable \cite{CamSaAlc13,Strass13}, L-Supported models \cite{Strass13}  &  GL-supported models (Theorem \ref{theorem:semi-models-semi-models})  & L-supported models, L-stable models \cite{Strass13} \\ \hline
Ideal Semantics   &       & $WFS^+$ \cite{NivOso14}  &  \\ \hline
CF2 Semantics   &       &    & $MM^*$ \cite{NivOsoZep11} \\ \hline
Stage Semantics   &       &  &  GL-stage models (Theorem \ref{theorem:stage-models})\\ \hline
\end{tabular}\label{table:summary}

\end{table}

Table \ref{table:summary} argues for a \emph{strong} relationship between argumentation inference and logic programming inference. Moreover, we can observe that the argumentation semantics which have been characterized by logic programming semantics have been studied from different points of view, \eg Labellings \cite{BaroniCG11}. This evidence argues that any \emph{well-defined argumentation semantics} must be characterized by a logic programming semantics.  However, further research is required in order to identify the necessary conditions which could support a basic definition of a \emph{Well-defined Non-monotonic Inference} of any argumentation semantics. These conditions can be identified in terms of \emph{non-monotonic reasoning properties} which have been explored in both fields argumentation and logic programming, \eg the property of \emph{relevance}  \cite{Cam06,NivOsoZep11}.

The exploration of argumentation as logic programming inference is not limited to the characterization of argumentation semantics in terms logic programming semantics. Since Dung's presented his seminal paper \cite{Dung95}, he showed that logic programming can support the construction of argumentation-based systems. Currently there are quite different logic-based argumentation engines which support the inference of argumentation semantics \cite{CharwatDGWW15,UweGagWol10,ToniSergot11}. It is well-known that the computational complexity of the decision problems of argumentation semantics ranges from NP-complete to $\Pi_2^{(p)}$-complete. In this setting, Answer Set Programming has consolidated as a strong approach for building argumentation-based systems \cite{CharwatDGWW15,UweGagWol10,ToniSergot11,NivOsoZepCor05}.

\section{Conclusions}

Currently, most of the well accepted argumentation semantics have been characterized as logic programming inference (Table \ref{table:summary}). This evidence argues that whenever a new semantics  appears, it is totally reasonable to search for a characterization of it as a logic programming inference.

According to Theorem \ref{theorem:semi-models-semi-models}, semi-stable semantics can share the same mapping (\ie $\Pi_{AF}$) with grounded, stable, preferred, complete and ideal semantics for being characterized as logic programming inference. This result argues that all these argumentation semantics can share the same interpretation of an argumentation framework as a logic program. Certainly, the logic programming semantics which are considered for characterizing these argumentation semantics share also a common interpretation of the argumentation inference which is restricted to the class of programs defined by $\Pi_{AF}$.
We have also showed that stage semantics can be also characterized by a logic programming semantics (Theorem \ref{theorem:stage-models}). This result argues that stage semantics has also logic programming foundations. Considering Theorem \ref{theorem:semi-models-semi-models} and Theorem \ref{theorem:stage-models}, we can give a positive answer to Q2.

An interesting observation, from the results of this paper, is that the concept of range which is fundamental for defining semi-stable and stage semantics can be captured from the logic programming point of view by considering $SC_1$ which can be based on well-acceptable reductions from logic programming. 
It is worth mentioning that reductions as GL and RED suggest some general rules for managing negation as failure. This evidence suggests that $SC_1$ defines an approach for answering $Q1$.

We argue that $SC_1$ suggests a generic approach for exploring the concept of range in two directions: to explore ranges as logic programming models and to explore new argumentation semantics based on both logic programming models and ranges.


\bibliographystyle{acmtrans}
\bibliography{papers_jcns}

\end{document}